\documentclass[review]{elsarticle}

\usepackage{lineno,hyperref}
\modulolinenumbers[5]

\journal{Journal of \LaTeX\ Templates}

\usepackage[dvipsnames]{xcolor}
\usepackage{url}
\usepackage{csquotes}

\usepackage{amsfonts,amssymb} 
\usepackage{mathtools,lipsum, nccmath}
\usepackage{graphics} 
\usepackage{mathptmx} 
\usepackage{times} 
\usepackage{cancel}
\usepackage{mathtools}

\usepackage{algorithm}
\usepackage{algorithmic}
\usepackage{lingmacros}
\usepackage{color}
\usepackage{float}
\usepackage{hyperref}
\usepackage{cleveref}
\usepackage{bm}
\usepackage{array}
\usepackage{tabu}

\hypersetup{
	colorlinks=true,
	linkcolor=blue,
	filecolor=blue,      
	urlcolor=blue,
	citecolor=blue
}

\usepackage{amsthm}

{

	\newtheorem{theorem}{Theorem}[section]

}
\let\oldproofname=\proofname
\renewcommand{\proofname}{\rm\bf{\oldproofname}}
\makeatletter
\renewenvironment{proof}[1][\proofname]{%
	\par\pushQED{\qed}\normalfont%
	\topsep6\p@\@plus6\p@\relax
	\trivlist\item[\hskip\labelsep\bfseries#1\@addpunct{.}]%
	\ignorespaces
}{%
	\popQED\endtrivlist\@endpefalse
}
\makeatother

\DeclareMathOperator{\Tr}{Tr}
\DeclareMathOperator{\col}{col}

\newcommand{\mbtheta}{\boldmath{\theta}}
\newcommand{\mbx}{\mathbf{x}}

\newcommand{\mbu}{\mathbf{u}}

\newcommand{\mbf}{\mathbf{f}}
\newcommand{\mbe}{\mathbf{e}}
\newcommand{\mbF}{\mathbf{F}}

\newcommand{\mbSigma}{\pmb{\Sigma}}
\newcommand{\mbLambda}{\pmb{\Lambda}}
\newcommand{\mbsigma}{\pmb{\sigma}}
\newcommand{\mbmu}{\pmb{\mu}}
\newcommand{\mbomega}{\pmb{\omega}}

\newcommand{\BoxD}{\textsf{Box2D}}

\let\oldmb\mathbold
\protected\def\mathbold{\oldmb}

\usepackage{graphicx}
\usepackage{verbatim}   
\usepackage{color}      
\usepackage{subfigure}  
\newcommand*{\rom}[1]{\expandafter\@slowromancap\romannumeral #1@}
\makeatother

\begin{document}
	%

\begin{frontmatter}
	
	\title{Reinforcement Learning Using Expectation Maximization Based Guided Policy Search
		for Stochastic Dynamics} 
	
%

	\author[mymainaddress]{Prakash Mallick}
	
	\author[mymainaddress]{Zhiyiong Chen\corref{mycorrespondingauthor}}
	\cortext[mycorrespondingauthor]{Corresponding author}
	\ead{zhiyong.chen@newcastle.edu.au}
\author[mymainaddress,mysecondaryaddress]{Mohsen Zamani}

	\address[mymainaddress]{School of Electrical Engineering and Computing, \\
		University of Newcastle, Callaghan, NSW 2308, Australia}
	\address[mysecondaryaddress]{Department of Medical Physics and Engineering,\\
	Shiraz University of Medical Sciences, Shiraz, Iran.}
	
	\begin{abstract}
	Guided policy search algorithms have been  proven to work with incredible accuracy for not only controlling a complicated dynamical system, but also learning optimal policies from various unseen instances. One assumes true nature of the states in almost all of the well known policy search and learning algorithms. This paper deals with a  trajectory optimization procedure for an unknown dynamical system subject to measurement noise using expectation maximization
	and extends it to learning (optimal) policies which have less noise because of lower variance in the optimal trajectories. Theoretical and empirical evidence of learnt optimal policies of the new approach is depicted in comparison to some well known baselines which are evaluated on an autonomous system with widely used performance metrics.
	\end{abstract}
	
	
\end{frontmatter}


	%

	\section{Introduction}

	Probabilistic inference in reinforcement learning has received increased interest among not only systems and control but also 
	artificial intelligence communities. Researchers are specifically {interested in} handling control related tasks on real scenarios as well as generalizing their learnt policies to new behaviors through experience. Recently, reinforcement learning has been  proven to work with incredible accuracy and success to provide a solution to these problems.

	{Reinforcement} learning is widely used for solving a Markov decision process (MDP) or partially observable Markov decision process (POMDP) by optimizing a reward function  while learning the intelligent decisions. Policy search has been carried out in reinforcement learning setting
	for robotic applications such as manipulation tasks \cite{pastor2009learning,deisenroth2011learning,chebotar2017combining} and game playing \cite{kober2011reinforcement}.  In particular, model-based policy search problems have been addressed using trajectory optimization \cite{zhang2016learning,levine2016end}, analytical policy gradients \cite{deisenroth2011pilco,tedrake2004stochastic}, and information-theoretic approaches \cite{deisenroth2013survey,williams2017information}, while 
	model-free policy search has also been studied in, e.g., controlling a robot \cite{vlassis2009learning}. Deep model-free reinforcement learning approaches have been making considerable progress in solving very high dimensional control problems \cite{lillicrap2015continuous,schulman2015trust} but their very high sample complexity is a hindrance for many practical applications. 
	Therefore,  guided policy search  (GPS) \cite{levine2013guided} was proposed to 
	address the challenge of sample complexity with high dimensionality by dividing it into two problems: 
	i) a local model-based trajectory optimization step to produce guiding (expert) policies; and ii) a supervised learning process that utilizes local optimal policies as a guide to train a high-dimensional policy neural network. As a result it can successfully generalize
	polices for unseen scenarios with relatively less samples.  
	
	In the bulky literature hovering around the GPS framework established by \cite{levine2013guided}, there have been numerous variants such as  path-integral GPS  \cite{chebotar2017combining},
	path-integral linear quadratic regulator based GPS \cite{chebotar2017path}, mirror-descent GPS \cite{montgomery2016guided}, 
 Bregman alternating direction method of multipliers (BADMM) based GPS \cite{levine2016end}, model predictive GPS \cite{zhang2016learning}, and so on. All of these 
	variants focus directly on choosing intelligent decisions, as well as handling the generalization across task instances relatively well. However, they do  not deal with uncertainties of latent states especially in the trajectory-centric optimization phase, and 
	as a result the learning and generalization performance is affected. For example, although sensor observations
	are utilized to make intelligent decision during testing, they rely on the full states to carry out the guiding step which is 
	a restrictive limitation \cite{zhang2016learning}. 
	
	The aforementioned observation has motivated policy search in a partially observable framework.
	The early work by \cite{porta2006point} addressed the optimal control problem for solving POMDPs with a linear
	Gaussian transition model and a mixture of Gaussian reward model. Nevertheless, the approach requires an action space to be discretized. This limitation was handled  by the approach in \cite{toussaint2006probabilistic}. Substantial amount of work has been done in the field of policy search using {expectation maximization} (EM). For instance, the work in \cite{cooper2013method} concentrates on inference for decision making by  utilizing likelihoods and rewards for solving the inference problems. The EM technique has also been used for approximate inference in model-free learning \cite{toussaint2006probabilistic,hoffman2009expectation}, where 
	the cumulative sum of the expected rewards are maximized by a reward proportional predictive distribution. 
	In addition to that, 
	\cite{dayan1997using} utilized the concept of likelihood for solving an optimal control objective in a binary reward setting
	and  \cite{vlassis2009learning} proposed  a model free approach, i.e., Monte-Carlo EM to learn complex tasks for a real robot.
	More specifically, the latter exploits importance sampling to weight the reward factors produced from sampling the trajectories and then maximizes the likelihood of observing higher rewards.
	The recent work in \cite{prakasharxiv} has successfully provided an optimal control framework for handling partially observable nature of states. Nevertheless, most/all of these strategies lack in addressing the issue of generalization and learning in a model-based partially observed scenario, which leads to the motivation of this paper.

	In this paper, we provide a novel variant of the GPS algorithm which utilizes a maximum likelihood (ML) based optimal control to carry out learning in the presence of uncertainties (specifically arising from latency in states).
	We leverage a robust numerical implementation of EM which has been extensively used in system identification and extend it towards learning and generalization from unseen initial {conditions}. A theoretical analysis of covariance matrix has been developed, that intuitively quantifies less noise in the learnt policies. The performance of the proposed approach is also investigated based on the sample efficiency and success of generalization from multiple testing instances. Furthermore, the paper leverages strong empirical results to justify the claim that the EM-based GPS approach outperforms some of the well known variants of existing GPS algorithms on a set of synthetic data.
	
	The paper is organized in the following manner. Section~\ref{sec:preliminaries} lays out some fundamentals including some preliminaries, controller parameter space, and some assumptions. Section~\ref{sec:problem} sheds light on the problem formulation. Then, Section~\ref{method} explains thoroughly the EM-GPS approach, i.e., obtaining local dynamics model, trajectory optimization and policy learning. This section also provides a theoretical result in terms of singular values which quantifies the noise in the learnt policies.  Section~\ref{sec:results_leanring} evaluates the experimental results based on some well defined metrics and compares them with three GPS benchmarks evaluated on a \BoxD\ framework.  Section~\ref{conclusion} concludes the paper with some future extensions.

	\section{Background} \label{sec:preliminaries}
	\subsection{Preliminaries}
	
	The paper considers a reinforcement learning framework in which
	an agent interacts with a complicated environment by 
	making intelligent decisions based on a predefined objective function. The interaction leads to 
	nonlinear stochastic dynamics which do not have a valid model from the first principle. 
	The complicated dynamical system is referred to be called in this paper as a  global model (O) composed of multiple local models 
	$o^l,\;  l = \{1, 2, \cdots\}$, each of which follows a structure as shown in Fig~\ref{fig:MDP}.
	We are interested in devising methodologies for finite-horizon optimal control and learning (for all possible initial states) in the presence of noise. 
	A real system in the presence of uncertainties such as parameter variation, external disturbance, 
	and sensor noise creates latency in the underlying states of the system which propagates into the control action through the unknown dynamical equations. So each of the local models (of the unknown dynamical equation) can be modeled as a POMDP. It has a {\it latent state} $\mathbf{x}_{k} \in \mathbb{R}^{n_x}$ 
	and a {\it control action} $\mathbf{u}_k \in \mathbb{R}^{n_u}$, for each time instant $k = 1, 2, \cdots$, and the 
	local state transitional dynamics is represented with a conditional probability density function (p.d.f.), i.e., 
	\begin{align} \label{modelpdf}
	p (\mathbf{x}_{k+1} | \mathbf{x}_k, \mathbf{u}_k).
	\end{align}
	Specifically, for $k=1$,  $\mathbf{x}_1 \in \mathbb{R}^{n_x}$ follows some {\it initial state} distribution which is assumed to be known. 
	We specifically consider a finite-horizon  POMDP throughout this paper  for $k = 1, 2, \cdots, K$, called as an 
	{\it episode}, with the end of episode denoted by time $K$ being the end of episode.
	
	The entity $Y_k(\mathbf{x}_k,\mathbf{u}_k) \in \mathbb{R}^+$ denotes the instantaneous real valued {\it running cost} for executing action $\mathbf{u}_k$ at a state $\mathbf{x}_k$. 
	Precisely, it can be defined in a quadratic manner as,
	\begin{align} \label{quad_reward}
	Y_k(\mathbf{x}_k,\mathbf{u}_k) = (\mathbf{x}_k-\mathbf{x}^*)^\top \mathbf{Q_x} (\mathbf{x}_k-\mathbf{x}^*) + (\mathbf{u}_k-\mathbf{u}^*)^\top \mathbf{Q_u} (\mathbf{u}_k-\mathbf{u}^*),
	\end{align} where $\mathbf{x}^*$ and $\mathbf{u}^*$ are the target state and 
	control action, respectively, and  $\mathbf{Q_x} > 0$ and $\mathbf{Q_u}> 0$ are some specified matrices. We assume the  cost function is known but the transition function is unknown.
	As $\mathbf{x}_k$ and $\mathbf{u}_k$ are random variables,  
	$ Y_k(\mathbf{x}_k,\mathbf{u}_k)$  (with $Y_k$ a continuous and deterministic function) is also a random variable, shorted 
	as $Y_k$.  We develop another variable, i.e., $y_k \in \mathbb{R}^+$ (known as \textit{observed cost}) which   is described by a p.d.f. 
	$ p (y_k ( \mathbf{x}_k, \mathbf{u}_k )) $ or shorted as $p (y_k )$. We know that both $\mbx_k$ and $\mbu_k$ follow a Gaussian distribution, 
	therefore $Y_k$ follows a linear combination of independent non-central chi-squared variables with some degrees of freedom. 
	We simply assume that the p.d.f. of $Y_k$  follows an exponential distribution with parameter $\lambda$, i.e.,
	\begin{align}\label{pYk}
	p(Y_k) = \lambda e^{-\lambda Y_k} \text{ where }  \lambda>1.
	\end{align}
	Some relevant discussion can be found in  \cite{prakasharxiv,cooper2013method,dayan1997using,toussaint2009robot}.
	We also give a specific definition of $y_k$ as follows, 
	\begin{align} \label{expo_transformation}
	y_k  = e^{-  Y_k } ,
	\end{align}
	which utilizes an exponential transformation and is widely used in the literature for inference in optimal control scenarios; see, e.g., \cite{toussaint2009robot,peters2006policy,prakasharxiv}.
	
	Overall, the POMDP consists of a transition dynamics $p(\mathbf{x}_{k+1}|\mathbf{x}_k,\mathbf{u}_k)$ and the observation p.d.f. $p(y_k|\mathbf{x}_k,\mathbf{u}_k)$, i.e.,
	\begin{align}   \label{dmodel}
	{ p \Big(  \begin{bmatrix}
		\mbx_{k+1}       \\
		y_k      
		\end{bmatrix} | \mathbf{x}_k, \mathbf{u}_k \Big) },
	\end{align}
	which is referred to as the {\it dynamics model} further throughout the paper. 
	
	\subsection{Controller Parameter Space} 
	This subsection presents the definition of parameter space of a controller that is utilized in the paper. The control action is sampled from a
	linear Gaussian p.d.f.  shown below,
	\begin{align}
	\pi_{{\theta_k}}({\mathbf{u}}_k|\mathbf{x}_k) = \mathcal{N} ( \mbF_k \mathbf{x}_k + \mbe_k ,{\mbSigma_k} ),    \label{control}
	\end{align}
	for some matrices $\mbF_k, \mbSigma_k$ and a vector $\mbe_k$, representing state feedback control. 
	The matrix $\mbSigma_k$ is symmetric positive definite, and
	$\mbSigma_k^{\frac{1}{2}}$ is the square root of ${\mbSigma_k}$ satisfying 
	${\mbSigma_k}  =(\mbSigma_k^{\frac{1}{2}}) ^\top \mbSigma_k^{\frac{1}{2}} $. 
	Let $\mbf_k =  \text{vec} (\mbF_k)$ and
	$\mbsigma_k =  \text{vec} (\mbSigma_k^{\frac{1}{2}} )$.
	Then,  the vector 
	
	\begin{align*} 
	{\theta}_k =\col (\mbf_k, \mbe_k, \mbsigma_k),
	\end{align*} is called the {controller parameter}. 
	Over the episode under consideration, the controller parameters are lumped as follows, 
	\begin{align} \label{eq:theta_def}
	\theta  =\col (\theta_1, \theta_2, \cdots, \theta_K) \in \Theta,
	\end{align} where $\Theta$ is some non-empty convex set of parameters which is a closed and bounded subset of $({n_u n_xK+n_uK} + {n_u n_uK} )$-dimensional Euclidean space.
	The nonlinear feature of the controller is represented by the variation of $\theta_k$ with $k$, which 
	aims to account for the nonlinear complexity of the  dynamical model. 
	It is also assumed that the running cost $Y(\mbx_k,\mbu_k)$ is bounded from above.

	\begin{figure}[t]
		\centering
		\includegraphics[scale=0.24]{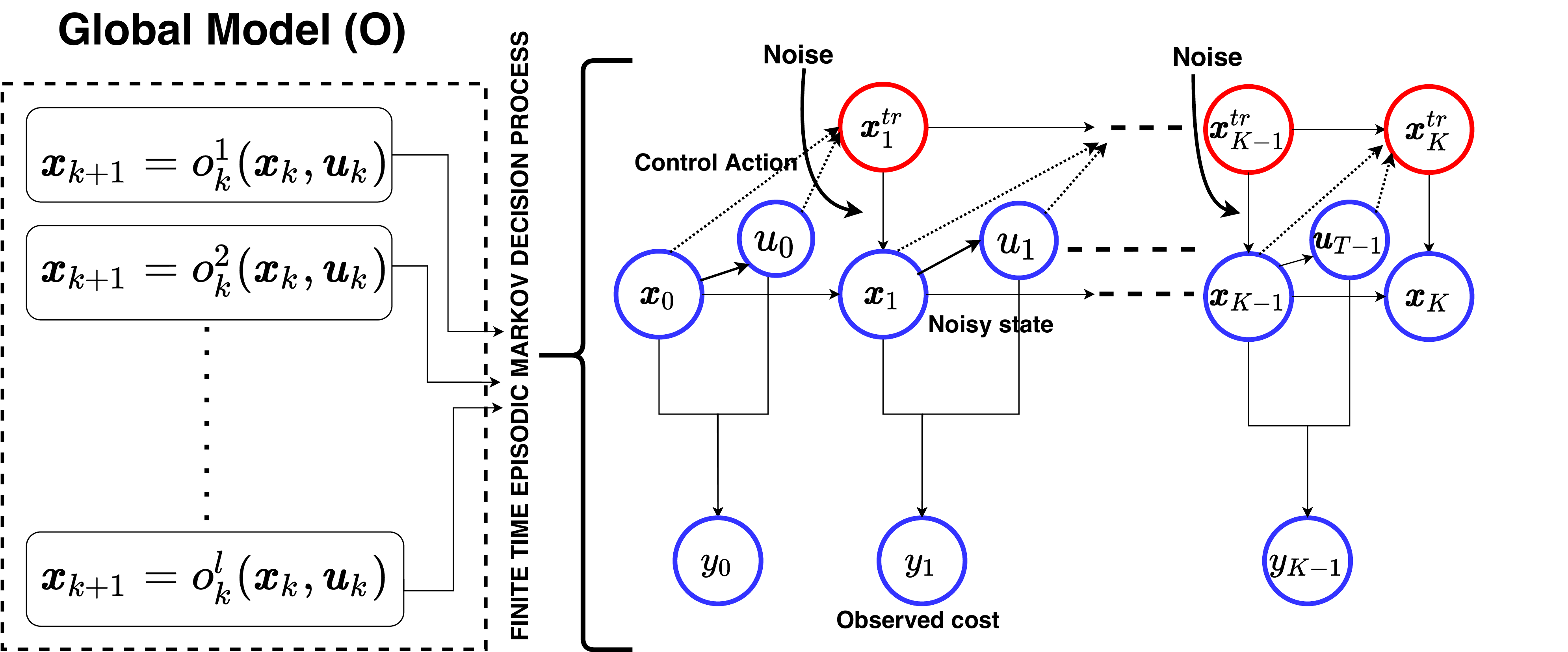}
		\caption{A global model divided into multiple local models and each of the local models follows an episodic MDP of $K$ steps. The red denotes the latent process and the blue denotes the observed process.}
		\label{fig:MDP}
	\end{figure}

	\section{Problem Statement} \label{sec:problem}
	The primary objective of this paper is to seek an optimal policy that achieves a stochastic optimal control objective for a system which is primarily subject to state uncertainties while learning the good policies which will serve as a guide for generalization from unseen initial conditions. Therefore we divide the objective briefly into two parts, i.e., Problems I and II.
	
	\subsection{Problem I}
	
	We first aim to solve for optimal control problem given as
	\begin{align} \label{ocp_action}
	{ \hat{\mbu}^*_{1:K}} = {\arg\min_{  {\mbu}_{1:K} }  \widetilde{V} ({\mbx}_{1},{\mbu}_{1:K}) }=  {\arg\min_{   {\mbu}_{1:K}}  \text{  } \mathbb{E} \text{  } [\sum_{k=1}^{K}  Y_k( {\mbx}_k, {\mbu}_k )}  ]
	\end{align}
	where $\widetilde{V} ({\mbx}_{1},{\mbu}_{1:K}) $ is the total sum of expected instantaneous costs $Y_k(\mbx_k,\mbu_k)$ under samples of all possible instantiations of next states arising as a result of noise in the system.
	Without loss of generality the objective function in terms of control action can be reduced to an optimization problem which solves for parameter $\mbtheta$ of the p.d.f. which governs the evolution of the control action sequence $  {\mbu}_{1:K} := \col (\mbu_1,\cdots,\mbu_K )$. 
	Therefore equivalently the stochastic optimal control objective function in terms of the parameter $\theta$ can be denoted as,
	\begin{equation} \label{obj_function_argmin}
	{ \hat{\theta}^*}= {\arg\min_{{\theta} \in \Theta}  \widetilde{V} ({\mbx}_1,{{\pi_\theta}}) }=  {\arg\min_{{\theta} \in \Theta}  \text{  } \mathbb{E} \text{  } [\sum_{k=1}^{K}  Y_k ( {\mbx}_k, {\pi_{\theta_k}} )} ].
	\end{equation}
	The term $\widetilde{V} ({\mbx}_1,{{\pi_\theta}})$ is the cumulative sum of the expected future returns and $\pi_{\theta_k} (\mbu_k |\mbx_k)$ is a non-stationary stochastic policy {parameterized by $\theta$} conditioned on state $\mbx_k$ as   \eqref{control}.
	Solving  \eqref{obj_function_argmin} is a global way of dealing with the stochastic optimal control problem and the globally optimal control law will be independent of starting initial state. However,  it is very hard to solve this kind of problems precisely. This is because in a reinforcement learning setting, solving POMDPs is theoretically proven to be NP-complete problem \cite{littman1994memoryless}. Therefore several approximations of value function have been developed in literature to tackle the complexity of the problem. For the rest of the paper, we are going to leverage one such analytical approximation that has been developed in \cite{prakasharxiv}.

	\begin{figure}[t]
		\centering
		\includegraphics[scale=0.25]{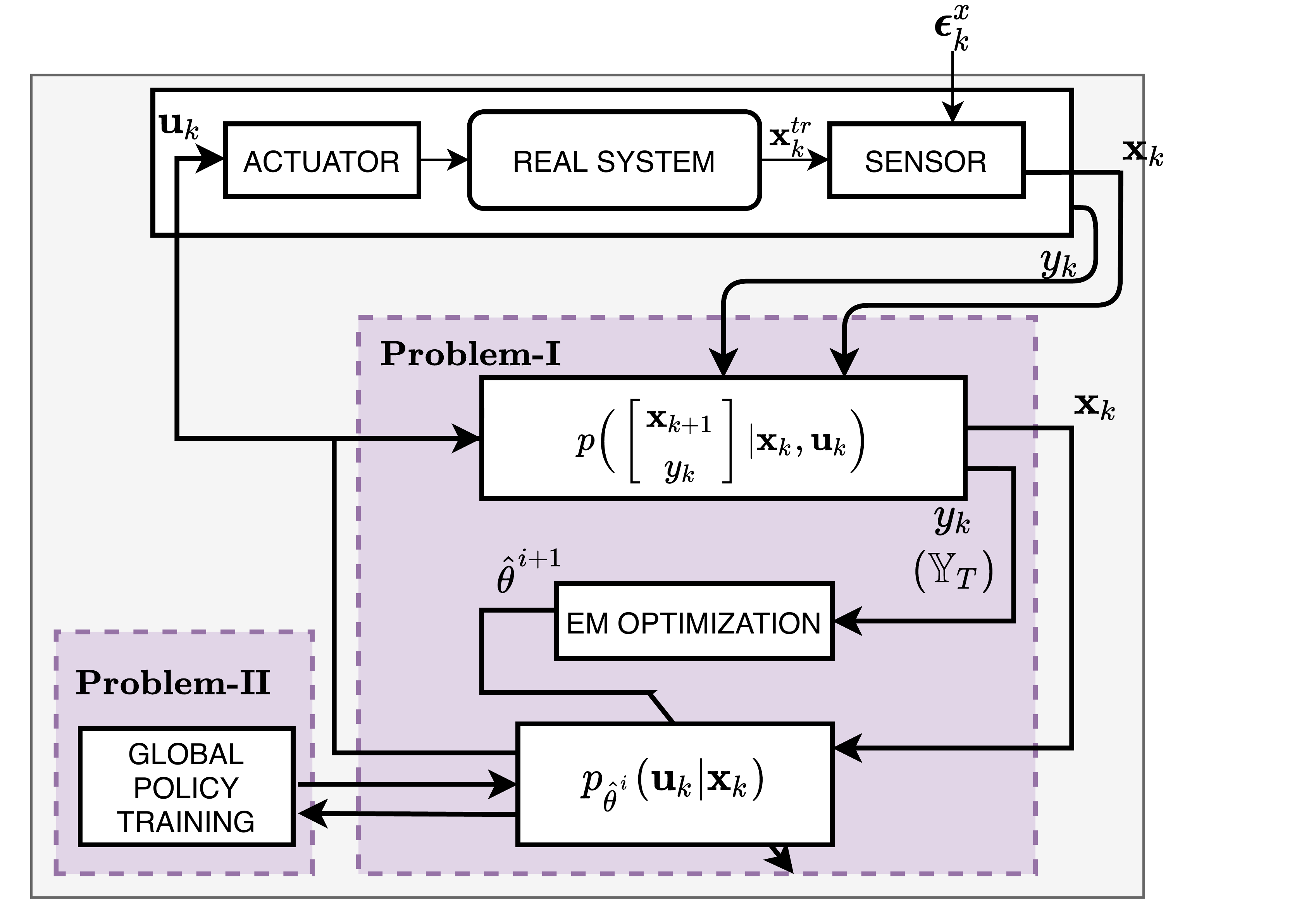}
		\caption{Block diagram representing the entire procedure of EM-GPS.}
		\label{fig:Blockdiag}
	\end{figure}
	
	In this approach, one deals with an objective function of the form \eqref{obj_function_argmin} but evaluated under a p.d.f which is obtained as a result of the EM algorithm (originally proposed by \cite{dempster1977maximum}). 
	The joint states vector, treated as {\it latent variables}, and  the cost vector, treated as {\it observations}, are shown as ${\mathbb{X}_{K+1}} \triangleq  \{{\mbx}_1,{\mbx}_2,\cdots,{\mbx}_{K+1}\}$ and ${\mathbb{Y}_{K}} \triangleq \{y_1,y_2,\cdots,y_{K}\}$, respectively. Then, 
	the log-likelihood of joint of  cost observations  and the log-likelihood of joint associated with cost and state pairs are represented as $ {  L_\theta ({\mathbb{Y}_K})} {\triangleq \log p_\theta(  {\mathbb{Y}_{K}} )}$ and ${ L_\theta ({\mathbb{X}_{K+1}} , {\mathbb{Y}_{K}}) } { \triangleq \log p_\theta(  {\mathbb{X}_{K+1}},{\mathbb{Y}_{K}}  )}$, respectively.
	We assume  $L_{{\theta} } (\mathbb{Y}_T)$ is bounded from above for ${\theta} \in \Theta$
	and the function $L_\theta(\cdot)$ is continuous in $\Theta$ and differentiable in the interior of $\Theta$.
	
	With these assumptions in mind, the stochastic optimal control version of the EM objective is studied in details in \cite{prakasharxiv}
	where  a {\it joint mixture likelihood} is optimized , i.e.,
	\begin{align}\label{jointlikelihood}
	\hat{\theta}^{i+1} = \arg\max_{\theta \in \Theta}  \mathcal{L}(\theta,\hat{\theta}^i),
	\end{align}
	for 
	\begin{align*}
	\mathcal{L}(\theta,\hat{\theta}^i) \triangleq    \mathbb{E}_{\hat{\theta}^i}   (  \log  p_\theta (\mathbb{X}_{K+1}, \mathbb{Y}_K) |\mathbb{Y}_K),
	\end{align*}
	where the parameter estimate $\hat{\theta}^i$ is some considerable parameter which one has knowledge of and 
	is recursively updated by increasing the likelihood.
	Throughout the paper, we use the simplified notation
	\begin{align*}
	\mathbb{E}_{\phi} ( * | \mathbb{Y}_K ) 
	\triangleq   {\mathbb{E}_{p_{\phi} (\mathbb{X}_{K+1}|\mathbb{Y}_K)} (  *) }.\end{align*}
	We refer to the problem of finding the optimal parameter $\hat{\theta}^*$  for  \eqref{obj_function_argmin}
	from the recursive approach \eqref{jointlikelihood}
	and in turn the optimal policy as Problem~I.  
	The result of Problem~I holds only for a limited set of initial conditions. Thus, to generalize the obtained optimal policy from {Problem} I, we switch the attention towards Problem~II.

	\subsection{Problem II}
	
	Problem~II deals with the issue of generalization, where the optimal parameters from the solution of Problem I are
	exploited and used to excite the real dynamical system to produce state marginals. We take a similar approach as taken by \cite{montgomery2016guided}, \cite{levine2013guided} and \cite{levine2016end} to utilize the samples of optimal (in the sense of optimization of Problem I) state marginals that act as guiding samples for the learning process.

	More specifically, we consider $C >1$ initial condition distribution $\mathcal{N}(\mathbf{s}^c,\mathbf{P}^c)$, $c=1,\cdots, C$.
	For each initial condition, the solution to Problem~I gives $\hat \theta_k^{i,c}$, recursively for $i=1, \cdots, I$.
	For each $\hat \theta_k^{i,c}$, the experiment repeats $S$ times with different initial states sampled from
	$\mbx_1^{i,c,s} \sim \mathcal{N}(\mathbf{s}^c,\mathbf{P}^c)$, for $s=1,\cdots, S$.
	For each state $\mbx_1^{i,c,s}$, the policy is
	\begin{align}
	{\mathbf{u}}_k^{i,c,s} \sim  \pi_{\hat{\theta}^{i,c}_k}({\mathbf{u}}_k|\mbx_k^{i,c,s}) = \mathcal{N} (
	\hat\mbF^{i,c}_k \mathbf{x}^{i,c,s}_k + \hat\mbe^{i,c}_k ,\hat{\mbSigma}^{i,c}_k ).  \end{align}
	Denote $\mbmu_k^{i,c,s} = \hat\mbF^{i,c}_k \mathbf{x}^{i,c,s}_k + \hat\mbe^{i,c}_k$. 
	For each iteration with $\hat{\theta}^{i}$, it gives
	$C S$  training samples  for the neural network  as follows,
	\begin{align}
	\{ \mathbf{x}^{i,c,s}_k,\;  \mbmu^{i,c,s}_k \}.
	\end{align}
	Then we utilize these samples in the next iteration to update the model, carry out the EM trajectory optimization and again generate samples from the state marginals which add more training data points to the previously trained neural network.

	Let $\zeta^L$ be the parameter vector of a Gaussian policy 
	\begin{align}
	\Pi_{\zeta^L} (\mbu_k |\mbx_k) \triangleq \mathcal{N}(\mbmu^L(\mbx_k), \mbSigma_k^L),
	\end{align} 
	where $\zeta^L$ represents the parameters of the  neural network generating $\mbmu^L(\mbx_k)$.
	Then, the parameter $ \mbSigma_k^L$ will be explicitly calculated and 
	the parameter $\zeta$  of the neural network is to be trained. Specifically, 
	the supervised learning objective function in the training process can be expressed as, for every $i$, 
	\begin{align} \label{closedformexpression}
	{\zeta}^{L,i} =  \arg\min_{\zeta}  \sum_{k=1}^K \sum_{c=1}^C \sum_{s=1}^S D_{KL} \Big( \Pi_{\zeta}
	({\mathbf{u}}_k|\mbx_k^{i,c,s}) || \pi_{\hat{\theta}^{i,c}_k}({\mathbf{u}}_k|\mbx_k^{i,c,s})  \Big) ,
	\end{align}  
	where $D_{\text{KL}} (\cdot)$ represents KL-divergence from the p.d.f. $\Pi_\zeta(\mbu_k|\mbx_k)$ to the p.d.f. $\pi_{\theta}(\mbu_k|\mbx_k)$. The term $\pi_{\hat{\theta}^{i,c}_k}({\mathbf{u}}_k|\mbx_k^{i,c,s}) $ represents the local control policy
	whose samples supervise/guide the global policy $\Pi_{\zeta}
	({\mathbf{u}}_k|\mbx_k^{i,c,s})$ which can be considered to be a neural network  . 
	
	Basically, there are two phases for the solution to \eqref{closedformexpression}. 
	One is to find the mean through an NN using standard supervised learning. 
	The  other is to find  the covariance estimate by utilizing identities of multivariate Gaussian and then setting the gradient of \eqref{closedformexpression} to $0$.

	Problems I and II together form the Expectation Maximization variant of the GPS approach, shorted as EM-GPS,
	which addresses the optimal policy search and learning when there is substantial impact of measurement noise. 
	This methodology can be graphically represented by the block diagram in Fig.~\ref{fig:Blockdiag}.
	The explicit EM-GPS process will be elaborated in the next section.  
	
	\section{The EM-GPS Approach} \label{method}
	
	The EM-GPS approach consists of two major steps with each corresponding to 
	the aforementioned two problems. 
  Before we introduce the two steps, 
		we need to explore a specific Gaussian p.d.f. for the POMDP model \eqref{dmodel} 
		using well established principles of system identification.

	\subsection{Dynamics Fitting and Cost Observation}

	The paper deals with a  locally time-varying linear model \eqref{dmodel}  of the form,
	\begin{equation} \label{eq3}
	{ p \Big(  \begin{bmatrix}
		{\mbx}_{k+1}       \\
		y_   k   
		\end{bmatrix} |  {\mbx}_k , {\mbu}_k  \Big) }   = \mathcal{N} \Big(  \mathbf{A}^P_k
	{ \begin{bmatrix}
		{\mbx}_{k}       \\
		{\mbu}_k       
		\end{bmatrix} }  ,   \pmb{\Sigma}^P_k  \Big).
	\end{equation}
 with the parameters constructed in a data-driven manner by fitting obtained datasets to the above equation.
	A variational-Bayesian (see e.g., \cite{bishop2006pattern}-Sec 10.2) approach, motivated by 
	\cite{khansari2010bm,levine2013guided}, is used here 
	to determine the prior for successful dynamics fitting.

	One can run one iteration of experiment and collect tuples of measured 
	$\{\mathbf{x}_k,\;\mathbf{u}_k,\;\mathbf{x}_{k+1},\;y_k\}$ for one episode $k=1,\cdots, K$. Practically, 
	one can repeat the experiments for $M$ times from the same initial conditions with a random seed value to gather sufficiently many samples, each of which is denoted by 
	\begin{align*}
	\mathcal{D}_k^m=\{\mathbf{x}_k,\;\mathbf{u}_k,\;\mathbf{x}_{k+1},\;y_k\}_{\text{$m$-th experiment}},
	\end{align*}
	for $m=1,\cdots,M$.  Let 
	$\mathcal{D}_k = \{\mathcal{D}_k^1,\cdots, \mathcal{D}_k^M \}$ 
	and $\mathcal{D} = \{\mathcal{D}_1,\cdots, \mathcal{D}_K \}$.
	Then, one can fit a Gaussian mixture model (GMM) to the data set  $\mathcal{D}_k$. 
	In particular, the VB inference method is used to determine the parameters of the GMM, i.e., the means, covariances and weights of the Gaussians. 
	
	The  GMM which is produced as a result of VB inference acts as a considerable global prior and it helps in bringing in information to construct a solitary normal-inverse Wishart (NIW) distribution. This NIW acts as a conjugate prior for a Gaussian distribution 
	\begin{align} \label{NIWp}
	p (\mathbf{x}_k,\mathbf{u}_k,\mathbf{x}_{k+1},y_k) = \mathcal{N} (\mbomega_k, \mbLambda_k).
	\end{align}
	Next, it will be elaborated that the NIW prior plays an essential role in attaining the parameters, i.e., the mean $\mbomega_k$ and the covariance ${\mbLambda}_k$.

	The procedure of fitting GMM to $\mathcal{D}$ involves constructing NIW distributions to act as prior for means and covariances of Gaussian distributions involved in mixture model. In addition to it, Dirichlet distributions are defined to be the prior on the weights of the Gaussian distributions which would explain the mixing proportions of Gaussians. {Then, iterative VB strategy is adopted} to increase the likelihood of joint variational distribution (see e.g., \cite{bishop2006pattern}-Section 10.2) for attaining the parameters of GMM. 
	The attained parameters of GMM are utilized to further obtain the parameters of the solitary NIW prior which acts  as a representative of the global GMM prior. The purpose of NIW prior is to garner the information contained in the global GMM prior.  The mean of the solitary NIW prior is $\mbomega^0_k = \sum_{f=1}^{F} ( w^f_k  \mbomega^f_k )$ where $w^f_k$ and $\mbomega^f_k$ are the weight and mean of the $f$-th  Gaussian in the GMM and ${F}$ is the total number of initialized Gaussian clusters.  The precision matrix $\mbLambda^0_k$ of the solo NIW prior is evaluated by calculating the deviation of each cluster from $\mbomega^0_k$. There are two more essential parameters of the solo NIW conjugate prior namely $n_0$ and $k_0$, which are set for the total $M$ samples. 
	Define the empirical mean, $\mbomega^{emp}_{k}\in \mathbb{R}^{n_u+2n_x+1}$ and the covariance $\pmb{\Lambda}^{emp}_k \in \mathbb{R}^{{(n_u+2n_x+1)}\times {(n_u+2n_x+1)}}$ as follows, for the data set $\mathcal{D}_k$, 
	\begin{align}
	\mbomega^{emp}_{k} &=\frac{1}{M}\sum_{m=1}^{M} {\mathcal{D}^m_k} 
	{\mbLambda}^{emp}_k = \frac{1}{M} {\sum_{m=1}^M}   ( \mathcal{D}^m_k  -{{\mbmu_k^{emp}}}) ( \mathcal{D}^m_k  -{{\mbmu_k^{emp}}})^{\top}.  \label{emp}
	\end{align}
	Next, one can carry out Bayesian update that results in a-posteriori estimates of the mean and precision matrix
	for the solo Gaussian in \eqref{NIWp}, that is, 
	\begin{align} \label{Bayesian}
	\mbomega_k=\frac{{ k_0  \mbmu^0_k} + M \mbmu^{emp}_k  }{k_0+M}  ,\;
	{\mbLambda}_k =   \frac{{ (\mbLambda^0_k)^{-1} } + M\cdot \pmb{\Lambda}^{emp}_k + \kappa_k }{M+n_0} 
	\end{align}
	where $\kappa_k =   [k_0 M / (k_0 + M)] ({\mbmu^{emp}_k}-\mbmu^0_k) ({\mbmu^{emp}_k}-{\mbmu^0_k})^{\top}$.

	The Gaussian distribution  \eqref{NIWp} can be conditioned on states and action, i.e., $(\mathbf{x}_k,\mathbf{u}_k)$, using standard identities of multivariate Gaussians, which delivers
	the following parameters of \eqref{eq3} i.e.,
	
	\begin{align*}
	\mathbf{A}^P_k = \begin{bmatrix}
	\mathbf{A}^d_k      & {\mathbf{B}^d_k} \\
	{\mathbf{A}^y_k}       & {\mathbf{B}^y_k} 
	\end{bmatrix} ,  \pmb{\Sigma}^P_k = {\begin{bmatrix}
		{\pmb{\Sigma}^d_k} & {{\pmb{\Sigma}^{yd}_k}} \\
		{\pmb{\Sigma}^{yd}_k}^\top      & { {\pmb{\Sigma}^y_k}}
		\end{bmatrix}  } .
	\end{align*}
 The dimensions of matrices are  $\mathbf{A}^d_k \in \mathbb{R}^{n_x \times n_x}$, ${\mathbf{B}^d_k} \in \mathbb{R}^{n_x \times n_u}$, ${\pmb{\Sigma}^d_k} \in \mathbb{R}^{n_x \times n_x}$, ${{\mathbf{A}^y_k}} \in \mathbb{R}^{1 \times n_x}$, ${\mathbf{B}^y_k} \in \mathbb{R}^{1 \times n_u}$, ${ {\pmb{\Sigma}^y_k}} \in \mathbb{R}$, $ \mathbf{A}^P_k \in \mathbb{R}^{({n_x + 1}) \times (n_u +n_x) }$ and ${ {\pmb{\Sigma}^P_k} }  \in \mathbb{R}^{(n_x+1) \times (n_x+1)}$.
	In the dynamical model \eqref{eq3}, the term $ {\pmb{\Sigma}^{yd}_k}$ denotes the correlation between $\mathbf{x}_{k+1}$ and $y_k$. Without loss of generality, one can assume that  ${\pmb{\Sigma}^{yd}_k}=0$. Note that one can also consider ${\pmb{\Sigma}^{yd}_k} \neq 0$ and utilize methods of de-correlation to carry out the entire procedure in a similar way. It is assumed that the covariance matrices are symmetric positive definite, that is, $ {\pmb{\Sigma}^d_k} > 0$,  ${ {\pmb{\Sigma}^y_k}}> 0 $, and  ${\pmb{\Sigma}^P_k}> 0$. It is also assumed that the pair $({\mathbf{A}^d_k}\text{, }{\mathbf{B}_k^d})$ is controllable 
	and $ {\mathbf{B}_k^d}^\top \mathbf{B}_k^d >0 $.

	Now, it is ready to propose the EM-based trajectory optimization policy
	and the policy learning strategy in the next two subsections, respectively. 
	
	%
	%

	\subsection{EM Optimization Policy - Problem I} \label{robust_optimization}
	
	This subsection aims to elaborate the method of  finding the optimal parameter $\hat{\theta}^*$  for  \eqref{obj_function_argmin} from the recursive approach \eqref{jointlikelihood} with an  initial known parameter estimates $\hat{\theta}^0$.
	The effectiveness of the approach has been extensively studied in \cite{prakasharxiv} based on the 
	relationship between the stochastic optimal control objective function in \eqref{obj_function_argmin}
	and the maximum likelihood objective function in  \eqref{jointlikelihood}.
	
	More specifically, 
	with the expectation carried out under samples of p.d.f. from a known parameter estimates $\hat{\theta}^{i}$,
	the iteratively updated $\hat{\theta}^{i+1}$ that increases the joint mixture likelihood function $\mathcal{L}(\theta,\hat{\theta}^i)$
	also decreases the (approximated) cost-to-go from some initial state $\mbx_1$, that is,    
	\begin{align*}
	\mathcal{L}(\hat\theta^{i+1},\hat{\theta}^i) \geq  \mathcal{L}(\hat\theta^{i},\hat{\theta}^i)  \implies V(\mbx_1, \pi_{\hat{\theta}^{i+1}}) \leq  V(\mbx_1, \pi_{\hat{\theta}^{i}})   ,
	\end{align*}
	where $V(\mbx_1, \pi_{{\theta} }) \triangleq \mathbb{E}_{\hat{\theta}^i}  [\sum_{k=1}^{K}  Y_k ( {\mbx}_k, {\pi_{\theta_k}} ) ]$ is an approximated surrogate function for $ \widetilde{ V }(\mbx_1, \pi_{{\theta} })$ \cite{prakasharxiv}.
	Attention is then turned towards evaluation of  the mixture likelihood $\mathcal{L}(\theta,\hat{\theta}^i)$. In fact, utilizing the Gaussian assumption associated with (6), 
	the likelihood function $\mathcal{L}(\theta, \hat{\theta^i})$ for the dynamical model  \eqref{eq3} and the controller \eqref{control} 
	can be evaluated by a time-varying linear Kalman filter and R.T.S. smoother components.  
	
	As a result, the proposed optimization paradigm aims to seek a better policy parameter  $\theta =\hat\theta^{i+1}$ for the next iteration
	than {$\theta =\hat\theta^{i}$} in the sense of maximizing (or increasing)
	$\mathcal{L} (\theta,\hat\theta^i)$, that is, 
	\begin{align} \label{optiphi}
	\hat\theta^{i+1}  = \arg\max_{\theta} \mathcal{L} (\theta, \hat\theta^i).
	\end{align}
	which, however, is typically difficult to compute. 
	Two practically effective methods were introduced in \cite{prakasharxiv}.

 We define a so-called information matrix  
	\begin{align} 
	\mathcal{I}  (\hat{\theta}^i, \hat{\theta}^{i+1}) &=
	\mathbf{I} -  \Big[ (-{ \nabla^2_{ \theta} \mathcal{L}(\theta, \hat{\theta}^i) \big|_{ \theta=  \hat{\theta}^{i+1}} })^{-1}    (-{  \nabla^2_{ \theta}  L_\theta (\mathbb{Y}_K) \big|_{\theta=\hat{\theta}^i}    } ). \Big]
	\end{align} 
	which can be utilized for analyzing the convergence of estimates of the EM algorithm. 
	Some particular interest is about the convergence for the covariance matrices.
	In the context of time-varying optimal control, the information matrix contains $\hat{\theta}^i = \col(\hat{\theta}^i_1,\cdots,\hat{\theta}^i_K)$ where $\hat{\theta}^i_k = \col(\hat{\mbf}^i_k ,\hat{\mbe}^i_k,\hat{\mbsigma}^i_k)$ and denoting $\hat{\mbsigma}^i = \col(\hat{\mbsigma}^i_1,\cdots,\hat{\mbsigma}^i_K)$.
	The principal minor of the $ \mathcal{I}  (\hat{\theta}^i, \hat{\theta}^{i+1})$  concerned with the covariance components 
	$\hat{\mbsigma}^i $ has the following inequality,
	\begin{align*}  
	\mathbf{0}    \leq \mathcal{I} _\Sigma (\hat{\theta}^i, \hat{\theta}^{i+1}) \leq \mathbf{I}.
	\end{align*}
	Then, the convergence of the recursive EM algorithm is assured with the following update law
	for the covariance matrices
	\begin{align}
	{\hat\mbSigma_k}^i  & =({\hat\mbSigma_k} ^{\frac{1}{2} i}) ^\top {\hat\mbSigma_k} ^{\frac{1}{2} i} ,\; \hat\mbsigma_k^i =  \text{vec} (\hat\mbSigma_k^{\frac{1}{2}i}) 
	\nonumber\\
	\hat\mbsigma^{i+1}  &=    \mathcal{I}_\Sigma  (\hat{\theta}^i, \hat{\theta}^{i+1}) \hat\mbsigma^{i}. \label{Sigmareduction}
	\end{align}

	\subsection{GPS Based on Supervised Learning - Problem II}
	
	In the proposed variant of GPS, we train the final nonlinear policy in a similar architecture as described in \cite{levine2016end,zhang2016learning},  but utilize the guiding samples of EM based optimal policies generated as a result of \eqref{jointlikelihood}.
	As the approach employs a similar supervised learning procedure, it inherits all the advantages of the GPS described in 
	the existing references such as \cite{levine2013guided,montgomery2016guided,chebotar2017combining,zhang2016learning}. 
	It also inherits the advantage of maximum likelihood strategies investigated  in \cite{prakasharxiv,gibson2005maximum} for  handling the latency in the states. 
	
	The policy optimization objective \eqref{closedformexpression} corresponds to minimizing the $\text{KL}$ divergence between global conditional policies $\Pi_\zeta(\mbu_k|\mbx_k)$ and the guiding optimized policies $\pi_{\hat{\theta}_k^{i}}(\mbu_k|\mbx_k)$.
	Because of the multivariate Gaussian formulation of the policies, the optimization can be expressed in a closed form 
	as follows, for $\zeta =\zeta^{L,i}$ at each iteration, 
	\begin{align}
	&  D_{KL} \Big( \Pi_{\zeta^{L,i}}
	({\mathbf{u}}_k|\mbx_k^{i,c,s}) || \pi_{\hat{\theta}^{i,c}_k}({\mathbf{u}}_k|\mbx_k^{i,c,s})  \Big) \nonumber \\
	= & \frac{1}{2} \log |\mbSigma_k^{L}|  -  \frac{1}{2} \Tr [({\hat{\mbSigma}^{i,c}_k})^{-1} 
	\mbSigma^L_k] \nonumber\\ & + \frac{1}{2} ({\mbmu}^{L,i}(\mbx_k^{i,c,s})- \mbmu_k^{i,c,s}) )^\top ({\hat{\mbSigma}_k^{i,c}})^{-1}  ({\mbmu}^{L,i}(\mbx_k^{i,c,s})- \mbmu_k^{i,c,s}) ). \label{eq:covar_learning}
	\end{align}
	Trained at each iteration based on the optimization objective \eqref{closedformexpression},
	the neural network is represented by the function  
	$\mbmu^{L,i}(\cdot)$. 
	As the final term of \eqref{eq:covar_learning} containing $\mbmu^{L,i}(\cdot)$
	is a weighted quadratic cost on the policy mean, it  can be learnt using standard supervised learning. 
	The aforementioned supervised learning approach inherits the foundation laid out in \cite{montgomery2016guided,chebotar2017combining}, etc. A major advantage of this approach can be attributed to
	the samples that come from a high-dimensional global policy while the (local) optimization  is carried out in a low-dimensional action space.

	It is noted that
	$ {\mbSigma}^L_k$ used in \eqref{eq:covar_learning} is fixed,  independent of the recursion $i$.
	Its expression is given below in a closed form by taking the derivative w.r.t. $\mbSigma^L_k$ and setting it to $0$ (see, e.g., \cite{montgomery2016guided}), i.e.,
	\begin{align}\label{eq:learning_covariance}
	{\mbSigma}^L_k = \left( \frac{1}{C}  \sum_{c=1}^{C}  ( \hat{\mbSigma}_{k}^{I,c}  )^{-1} \right)^{-1}.
	\end{align}
	Therefore the covariances of the global policy is updated by averaging the covariances of the local policies over the rollouts and initial conditions. In practice,  a simple diagonalization transformation works empirically well according to 
	\begin{align} \label{covariance_transformation}
	{\mbSigma}^L_k= 
	\left( \frac{1}{C}  \sum_{c=1}^{C}   ( \mbox{diag} (\sigma(\hat{\mbSigma}_{k}^{I, c} )))^{-1} \right)^{-1}.
	\end{align} 
	%
	%
	The next theorem provides a unique relationship between the adapted covariance matrix ${\mbSigma}^L_k$ obtained 
	from \eqref{covariance_transformation} and the covariance estimate $\hat{\mbSigma}_k^i$. 
	The result explicitly provides an evidence that the produced samples contain less noise, which means more action certainty in the excited dynamic system.
	
	%
	
	\begin{theorem} \label{thm:learning}
		Consider ${\mbSigma}^L_k$ satisfying  \eqref{covariance_transformation} where every 
		$\hat{\mbSigma}_{k}^{i, c}$, $i=1,\cdots, I$, is achieved by \eqref{Sigmareduction} with an initial $\hat{\mbSigma}_{k}^{0, c}$. 
		Then, 
		\begin{align*}
		\sum_{k=1}^K \Tr({\mbSigma}^L_k)^{-1}    &\geq   \frac{n_u^2 K^2}{C}   \sum_{c=1}^C   \left(
		{\sum_{k=1}^K  \Tr \hat{\mbSigma}^{0, c}_{k}} \right)^{-1}
		\end{align*}

	\end{theorem}
	
	\begin{proof}
		Let $\sigma^{i,c}_{k,q}$, $q=1,\cdots, n_u$, be the eigenvalues of 
		$\hat{\mbSigma}^{i, c}_{k}$, for $i=0$ and $i=I$.
		From \eqref{covariance_transformation}, one has
		\begin{align*}
		\sum_{k=1}^K \Tr({\mbSigma}^L_k)^{-1} &= \frac{1}{C} \sum_{k=1}^K  \sum_{c=1}^C \Tr \mbox{diag} (\sigma (\hat{\mbSigma}^{I, c}_{k})))^{-1}  
		= \frac{1}{C}   \sum_{c=1}^C  \sum_{k=1}^K \sum_{q=1}^{n_u} \frac{1}{\sigma^{I,c}_{k,q}}.
		\end{align*}
		By the arithmetic-harmonic mean inequality, i.e., for $a_i >0$, 
		\begin{align*} \sum_{i=1}^N \frac{1}{a_i}
		\geq \frac {N^2}{\sum_{i=1}^N a_i}, \end{align*}
		one has the following result
		\begin{align*}
		\sum_{k=1}^K \Tr({\mbSigma}^L_k)^{-1} 
		& \geq \frac{1}{C}  \sum_{c=1}^C   \frac{n_u^2 K^2}{\sum_{k=1}^K \sum_{q=1}^{n_u} \sigma^{I,c}_{k,q}} .
		\end{align*}
		Then,  by Theorem VI.4 of \cite{prakasharxiv}, one has
		\begin{align*}\sum_{k=1}^K \sum_{q=1}^{n_u} \sigma^{I,c}_{k,q} \leq \sum_{k=1}^K \sum_{q=1}^{n_u} \sigma^{0,c}_{k,q}
		\end{align*}
		and hence
		\begin{align*}
		\sum_{k=1}^K \Tr({\mbSigma}^L_k)^{-1}    &\geq \frac{1}{C}  \sum_{c=1}^C    \frac{n_u^2 K^2}{\sum_{k=1}^K  \sum_{q=1}^{n_u} \sigma^{0,c}_{k,q}} 
		= \frac{n_u^2 K^2}{C}    \sum_{c=1}^C     \frac{1}{\sum_{k=1}^K  \Tr \hat{\mbSigma}^{0, c}_{k} }. 
		\end{align*}
		The proof is completed.  
	\end{proof}

	\section{Simulation Results and Analysis}  \label{sec:results_leanring}
	
	This section describes the empirical results on synthetic simulation data.
	The algorithm based on EM based trajectory optimization and training of a high dimensional policy network for a global initial 
	state space has been extensively discussed. Its performance is evaluated on some common metrics with respect to 
	three existing variants of GPS in the references \cite{chebotar2017combining}, \cite{chebotar2017path}, and  \cite{montgomery2016guided}. The results are reported in this section.


	\subsection{Evaluation Platform and Baselines Methods}
	
We conduct the experiments on a benchmark GPS  \BoxD\ framework which describes a two-dimensional 
autonomous  system of second order point-mass dynamics where the ultimate goal is to be driven from any random initial condition to a target $[5, 20]$ in the Cartesian coordinate where the object is subject
		to gravity and linear damping. The training is done based on samples from two initial conditions i.e., $[0,5]$ and $[2,5.5]$ and we want the algorithm to learn and generalize to other test initial conditions (possibly in the neighborhood of the trained initial conditions). The state space of the system
	consists of the position and velocities, i.e., $\mbx = [x, y, \dot{x} , \dot{y} ]^\top \in \mathbf{R}^4$, and the control action $\mbu=[u_1,u_2]^\top \in \mathbf{R}^2$ denoting the acceleration profile. We deliberately add noise to the states at the start of the experiment, so as to finally verify the efficiency of the learning process in terms of both noise resilience and target reachability. The added noise is quantified as $\mbx_k=\mbx_k^{tr}+\pmb{\epsilon}_k^x$ where $\mbx_k^{tr}$ is the true state,  $\pmb{\epsilon}_k^x \sim \mathcal{N}(0,\mathbf{N}^x_k)$ and $\mathbf{N}_k^x = \varpi \mathbf{I}_{n_x}$. We used a noise factor of $\varpi=0.3$ throughout the simulations and  also used parameters of $\mathbf{Q}_x=\mathbf{I}_{4}$ and $\mathbf{Q}_u=5\times 10^{-5} \mathbf{I}_{2}$ in \eqref{quad_reward}. The other parameters are $K=30$, $S=10$, $C=2$, and $I=9$.

	We apply the proposed EM-GPS approach and  run simulations from numerous random initial conditions.
	Three {\it baselines} from the existing GPS variants are used in this section, i.e., ${\textbf{GPS:PI}}$ (with PI$^2$ trajectory optimization \cite{chebotar2017combining}), 
	${\textbf{GPS:iLQG}}$ (with an iLQR-based trajectory optimizer to enhance the $\text{PI}^2$  \cite{chebotar2017path}), and ${\textbf{GPS:MD}}$ (with the search as approximate mirror descent \cite{montgomery2016guided}). Using them as the initial parameter estimates $\hat{\theta}^0$, 
	the EM algorithm is applied, which results in three enhanced versions of GPS, i.e., 
	${\textbf{EM-GPS:PI}}$,  ${\textbf{EM-GPS:iLQG}}$, and ${\textbf{EM-GPS:MD}}$.


	%
	
	For each baseline, we first run the GPS algorithm until convergence is achieved and store the optimal parameters of $\hat{\theta}^0$ and simultaneously learn the neural network model parameter ${\zeta}^{L,0}$. The neural network policy is implemented on a deep learning Caffe based framework \cite{jia2014caffe}. Then, we carry out an offline computation of the EM optimization routine to generate the optimal policy parameters $\hat{\theta}^{i}$ and hence 
	train the neural network model parameter ${\zeta}^{L,i}$, for $i=1,\cdots, I$. 
	Finally, the neural network is tested for 10 new initial state distributions and
	10 samples from each distribution (100 experiments)  to evaluate the performance. 
	
	
	
	\subsection{Neural Network and Computation}
	
	The neural network employed for the policies consists of two fully connected hidden layers, each with $42$ dimension rectified linear (RELU) units \cite{nair2010rectified}. The neural network training is carried out for each iteration of the EM-GPS approach, with ADAM \cite{kingma2014adam} as the optimizer with the learning parameters as mentioned in \cite{zhang2016learning}. Inside each iteration of GPS, training is done for $4000$ epochs and inside each epoch with a minibatch size of $25$ and total batches per epoch of $50$.

	The offline EM optimization was conducted
	using multiple 2.6 GHz Intel Xeon Broadwell (E5-2697A v4)
	processors on the high performance computing (HPC) grid located
	at the University of Newcastle. 
	The other steps of the simulation were performed on Dell Alienware 15 R2 using a 64-bit on
	Ubuntu 16.04 OS, 8 Intel Core i7-6700HQ CPU @ 2.60GHz. 
	We leveraged the parallel processing power of the same in order to fasten the optimization.  We used Python 2.7 version for testing our offline optimization and online learning. During the offline maximization Scipy 0.10.0 was utilized. 

	\subsection{Evaluation of EM-GPS}
	{\it Cost-to-go: } The cost-to-go of the real costs $\sum_{k=1}^K Y_k(\mbx_k,\mbu_k)$ of the learned policies 
	for the aforementioned 100 testing experiments is evaluated for the three baselines and the three EM enhanced policies. The results 
	are visualized in terms of a box-plot which clearly depicts lower magnitudes of cost-to-go of the EM-based policies 
	compared to the baselines; see Fig~\ref{fig:cost}.
	
	\begin{figure}[t]
		\centering
		{\includegraphics[scale=0.3]{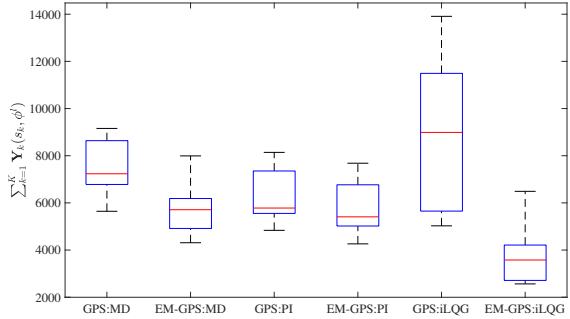}}
		\caption{Profile of cost-to-go of the learned policies for the three baselines and the corresponding EM enhanced policies.}
		\label{fig:cost}
	\end{figure}
	
	\begin{figure}[t]
		\centering
		{\includegraphics[scale=0.35]{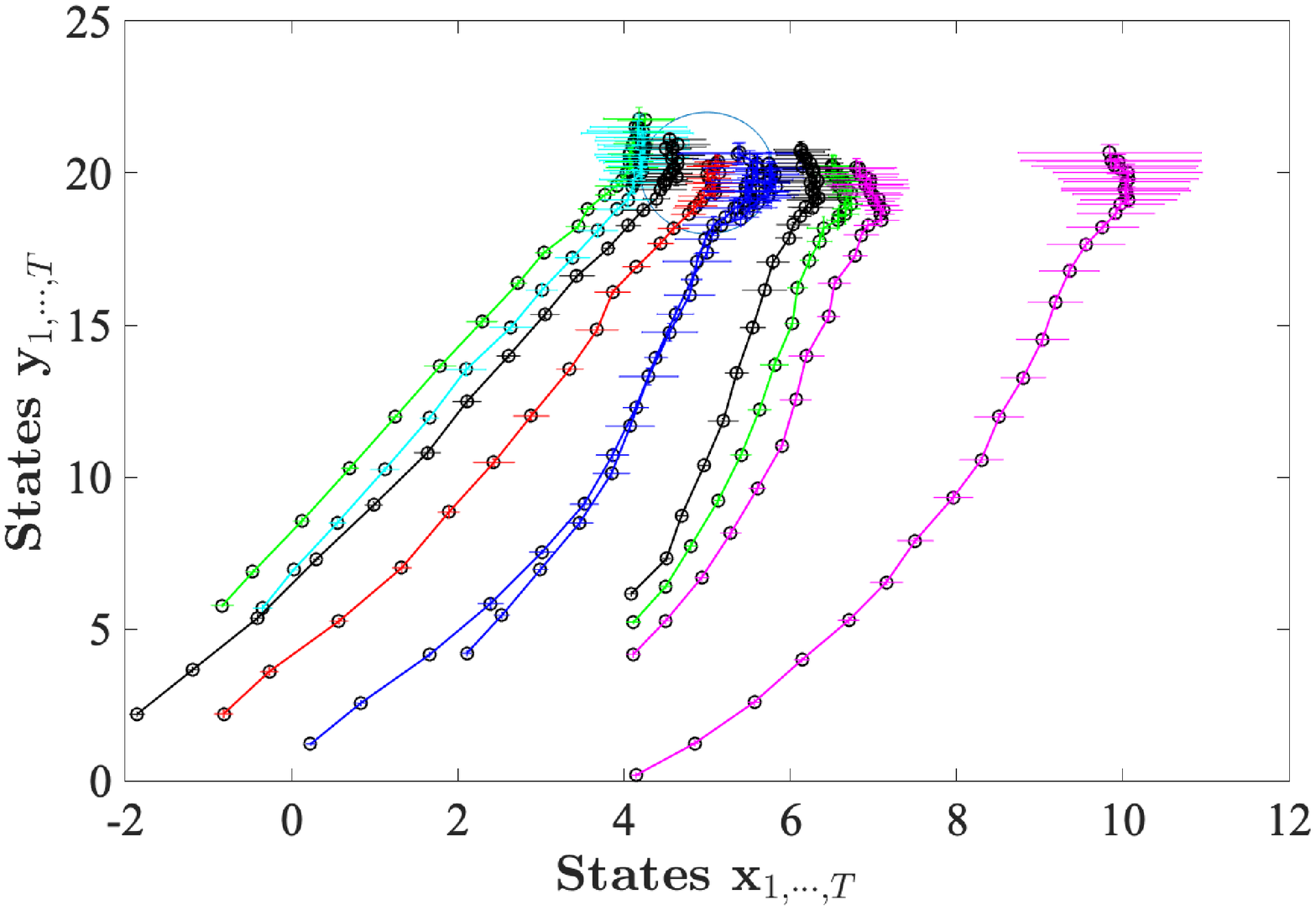}}
		\hfill
		{\includegraphics[scale=0.35]{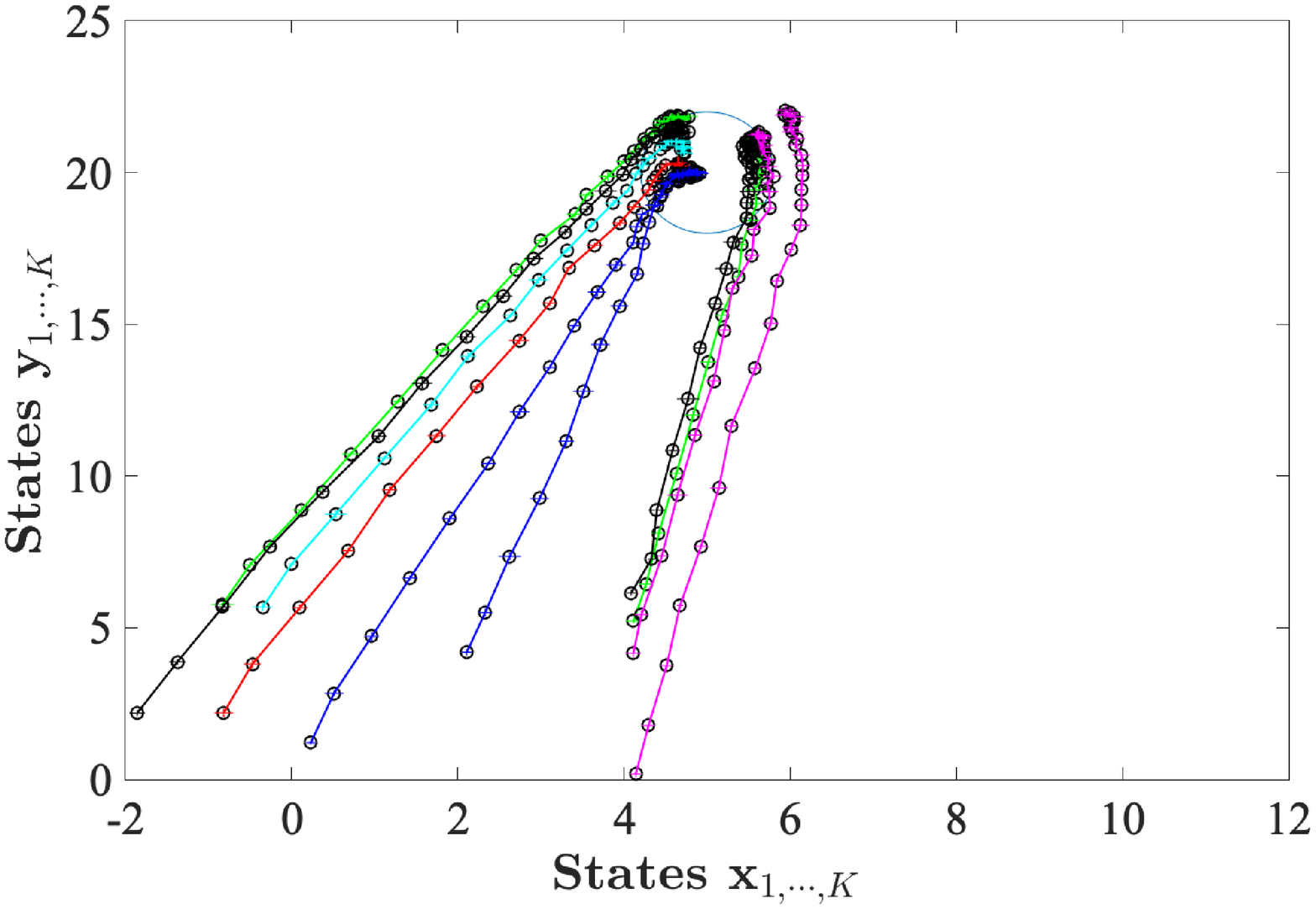}}
		\caption{Trajectories of states ${x}$ and ${{y}}$ for $10$ testing initial state distributions 
			and $10$ rollouts from each distribution     
			with $\textbf{GPS:PI}$ (top) and $\textbf{EM-GPS:PI}$ (bottom).}
		\label{fig:states}
	\end{figure}
	
	\begin{figure}[t]
		\centering
		{\includegraphics[scale=0.35]{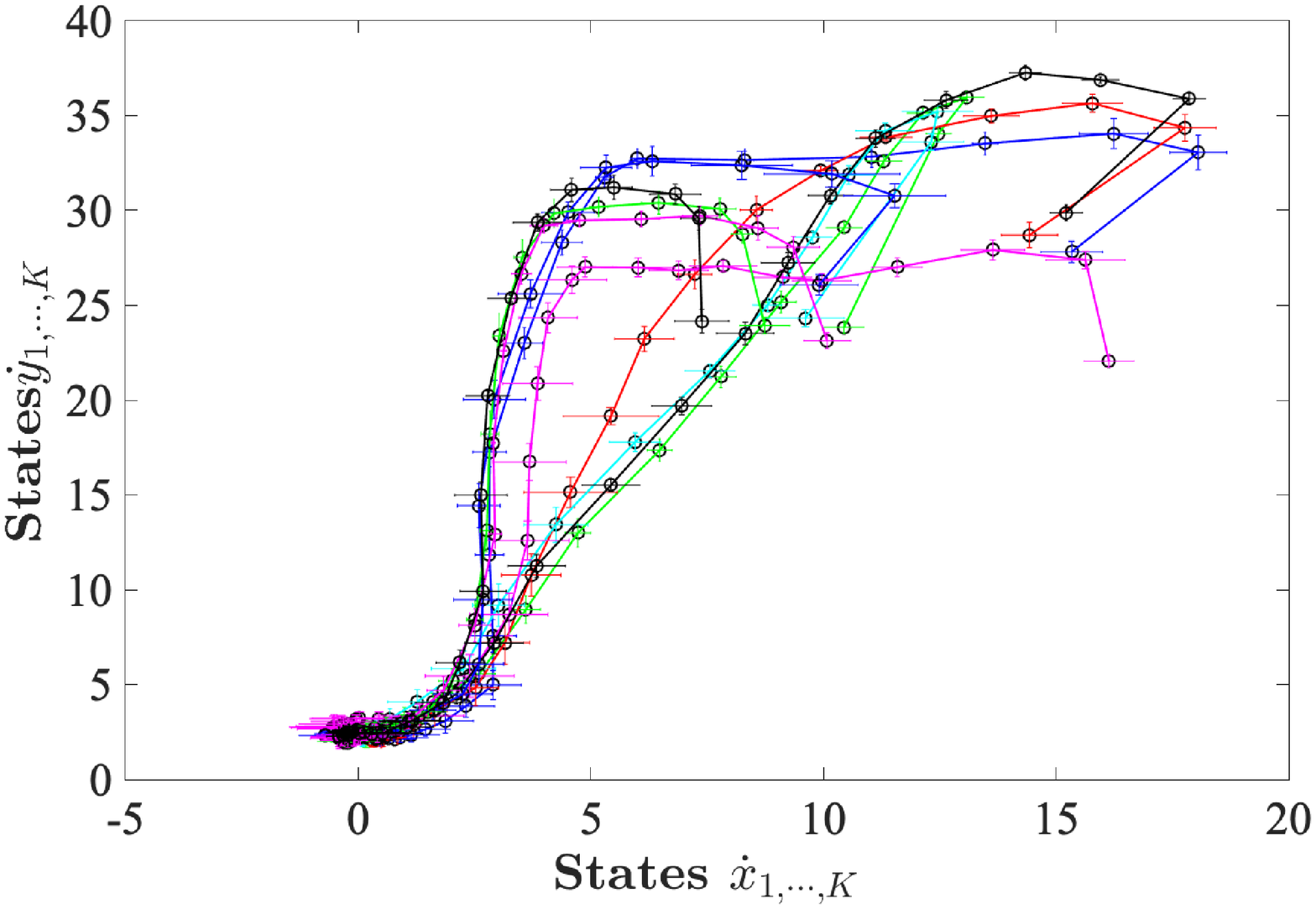}}
		\hfill 
		{\includegraphics[scale=0.35]{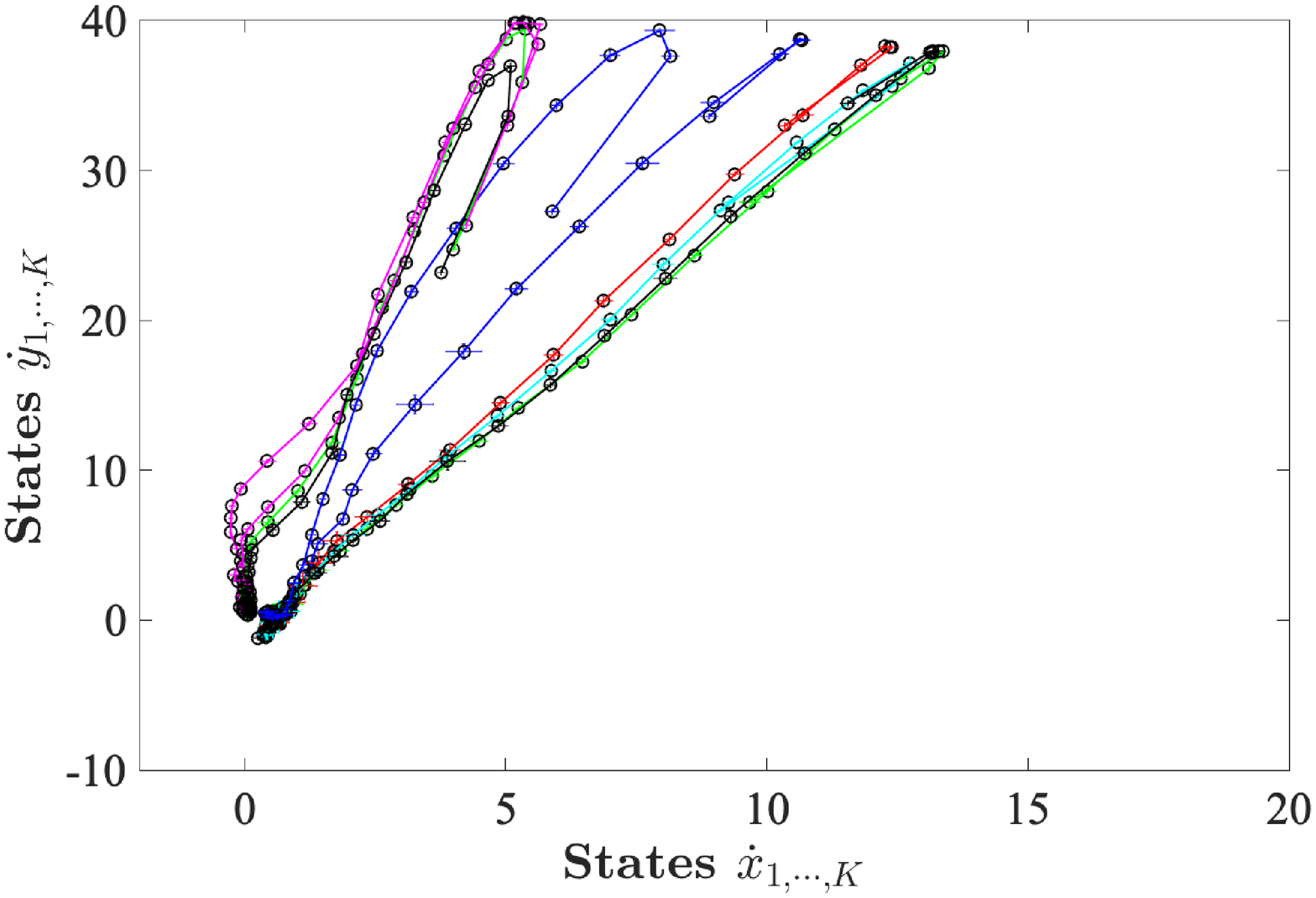}}
		\caption{Trajectories of states $\dot{x}$ and $\dot{{y}}$ for $10$ testing initial state distributions and $10$ rollouts from each distribution   with $\textbf{GPS:PI}$ (top) and $\textbf{EM-GPS:PI}$ (bottom).}
		\label{fig:velocity}
	\end{figure}
	
	{\it State trajectories:} 
	The true position trajectory $[x,y]$ of the \BoxD\ object is plotted  in Fig~\ref{fig:states} for the 100 testing experiments
	with the ${\textbf{GPS:PI}}$ baseline and the ${\textbf{EM-GPS:PI}}$ policy, respectively. In particular,  
	$10$ randomly chosen initial distributions around the neighborhood of $[0,5]$ and $[2,5.5]$ were selected 
	and $10$ rollouts were taken from each distribution. The mean$\pm$std-dev of the position trajectory is depicted in these figures. 
	It is observed that most of the generalization by the path-integral ($\text{PI}^2$) trajectory optimizer was unsuccessful in reaching the final goal target. On the contrary, the mean$\pm$std-dev of the samples obtained from the EM approach exhibited 
	better performance in reaching the target coordinate. 
	Not only the generalization from different task instances of the EM policy was better in terms of the trajectories, but also it is evident that the accumulation of the uncertainties in the baseline trajectories was more noisy as compared to that of EM. 
	It verifies the learning efficiency of the EM-GPS approach in handling noise.
	Additionally, the velocity trajectories $[\dot{x},\dot{y}]$ are plotted in Fig~\ref{fig:velocity} and they are expected to approach $0$ corresponding to staying still when the object reaches the target. Similar observation supports the success of the EM-GPS approach.

	{\it Success of generalization: } We define an ellipse centered the final target position $[x^*,y^*]= [5,20]$ for state trajectories with the elliptical x-radius of $0.8$ and y-radius of $2.0$.
	Also we define another ellipse for the acceleration profile, i.e., $[u_1,u_2]$,  centered $[0,0]$ with x-radius of 0.4 and y-radius of 1.5.
	We call an experiment is successful if the final position $[{x}^{tr}_K,y^{tr}_K]$ is within the first ellipse
	and the control actions ${\mbu}_{K-1}$ within the second one.
	Then, one can evaluate the number of successful experiments to verify the effectiveness of the generalization by GPS learning. Fig~\ref{fig:generalization} shows the success statistics for all the three baselines and the corresponding EM-GPS approaches
	using different learning samples. It can be clearly concluded that the quality of learning is comparatively better for EM-GPS as compared to most (almost all) of the trajectories from learning using the baseline policies. Also, it is clearly evident that the EM-based learning approaches are more sample efficient as compared to the baselines.
	
	\begin{figure} [t]
		\centering
		{\includegraphics[scale=0.4]{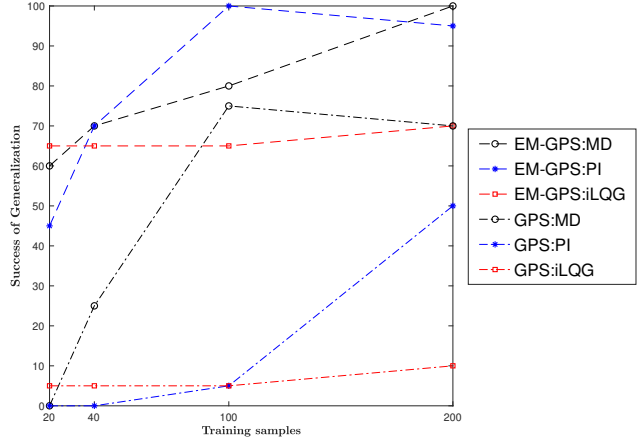}}
		\caption{Statistics of successful cases of the three baselines and their corresponding EM enhanced policies
			on 10 testing initial state distributions and 10 samples from each distribution (100 experiments).}
		\label{fig:generalization}
	\end{figure}

	{\it Control actions: }The control actions of the \BoxD\ object represents a proportionality to the force (i.e., acceleration profile) applied on the mass in both $x$ and $y$ direction. Fig~\ref{fig:KDE_plots} shows the kernel density plots of control actions on $5$ different initial state distributions with  5 rollouts for each distribution for each specific time step (we use less experiments here to keep the plot neat). Each cluster represents each time step and there are a total of $30$ clusters for each initial condition.  It clearly indicates that the optimal control actions as a result of the learning from ${\textbf{EM-GPS:PI}}$ contain substantially lesser noise as compared to the learned policies of ${\textbf{GPS:PI}}$.

	\begin{figure}[t]
		\centering
		{\includegraphics[scale=0.207]{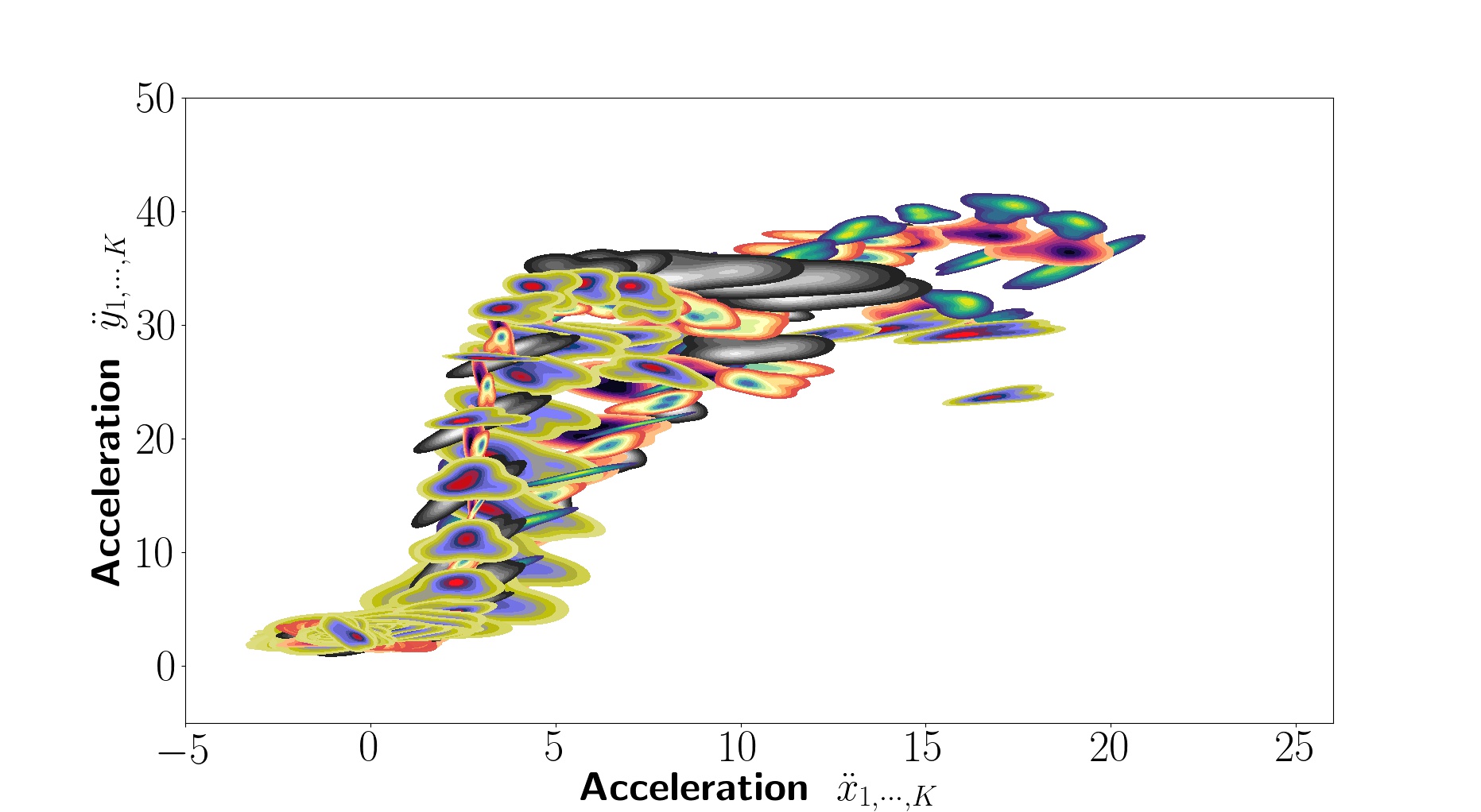}}
		\hfill
		{\includegraphics[scale=0.207]{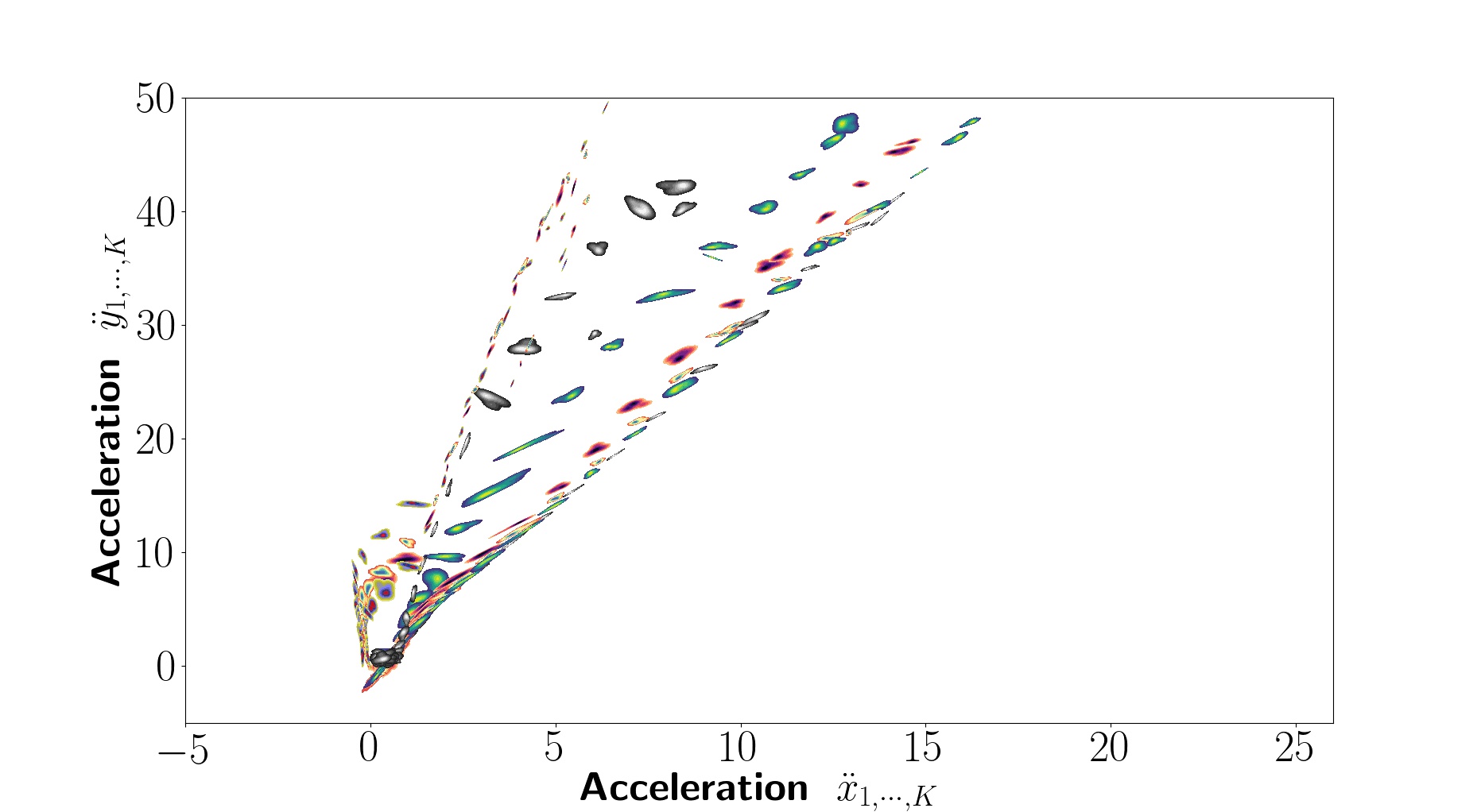}}
		\caption{Kernel density of the control actions for  $5$ testing initial state distributions 
			and $5$ rollouts  from each distribution represented by $5$ different colormaps, with $\textbf{GPS:PI}$ (top) and $\textbf{EM-GPS:PI}$ (bottom).}
		\label{fig:KDE_plots}
	\end{figure}


	\section{Conclusion} \label{conclusion}

	In this paper, we have proposed an EM-GPS approach which relies on considerable local policy parameter estimates of existing state-of-art trajectory optimization algorithms.
	It is a new variant of global policy learning approach which works effectively in a POMDP setting. We have provided theoretical analysis which states that optimal samples generated from neural network contain less noise as compared to the baselines. Also, we have obtained extensive empirical results on numerous performance metrics such as 1) approximated cost-to-go as the objective function; 2) respective error bars of state and action trajectories of learnt policies; 3) success of generalization from different random initial conditions. We have also shown that a trajectory produced as a result of those parameters 
	exhibits an efficient behavior in terms of reachability near the target.  
	The future work would specifically investigate two aspects of the proposed
	EM-GPS approach. On one hand, it is interesting to extend the approach of global policy learning to 
	nonlinear dynamical models where we may use a particle filter for a new  
	GPS variant. On other hand, we will further study stable dynamical models 
	by considering state and input constraints and seek approximation of the likelihood function that 
	may act as a cost function.

\bibliography{mybibfile}

\end{document}